\documentclass[journal]{IEEEtran}

\usepackage{amsmath,amssymb,amsthm,mathtools,cite}
\usepackage[hidelinks]{hyperref}
\hypersetup{
  pdftitle={Extended Proofs for Encrypted Sensing in Bistatic Radar: Unified Analysis and Randomly Activated Arrays},
  pdfauthor={Tianrui He, Ziheng Zheng, Guozheng Sun, Tianyao Huang, and Yimin Liu}
}

\newtheorem{proposition}{Proposition}
\newcommand{\E}{\mathbb{E}}
\newcommand{\Var}{\operatorname{Var}}
\newcommand{\Cov}{\operatorname{Cov}}

\title{Extended Proofs for Encrypted Sensing in Bistatic Radar:
Unified Analysis and Randomly Activated Arrays}
\author{Tianrui He, Ziheng Zheng, Guozheng Sun, Tianyao Huang, and Yimin Liu%
\thanks{T. He, Z. Zheng, G. Sun, and Y. Liu are with the Department of Electronic
Engineering, Tsinghua University, Beijing 100084, China.}%
\thanks{T. Huang is with the School of Computer and Communication Engineering,
University of Science and Technology Beijing, Beijing 100083, China.}}

\begin{document}
\maketitle

\begin{abstract}
This supplementary material provides the complete mathematical development
for the signal-to-noise ratio loss (SNRL) analysis of pulse-wise spatially
randomized encrypted sensing. We first derive a reduced model in which the
transmit-side randomization is represented by an independent and identically
distributed complex spatial coefficient. Based on this model, we establish
the almost-sure and mean asymptotic SNRL, derive the leading bias and variance
for a finite coherent processing interval (CPI), and obtain closed-form
coefficient moments and asymptotic SNRL for independent Bernoulli element
activation. These results complement the concise statements and proof sketches
in the associated paper.
\end{abstract}

\section{Introduction}

In a bistatic radar, the transmitter and receiver are spatially separated.
Besides the legitimate receiver, an illegitimate receiver in the illuminated
region may observe both the direct-path transmission and the target echo. The
direct-path signal can then serve as a reference for non-cooperative pulse
compression and coherent integration, as in passive and bistatic radar
processing \cite{griffiths2010klein,palmer2013dvb}. Encrypted sensing seeks to
prevent the illegitimate receiver from exploiting the illumination for target
sensing while preserving coherent processing at the legitimate receiver.

Pulse-to-pulse agility has been extensively studied in radar waveform and
array design; one representative example jointly varies carrier and spatial
resources in a multicarrier agile phased array \cite{huang2020multicarrier}.
For encrypted sensing, however, random changes that appear identically in the
direct path and the target echo can be estimated or removed by reference-based
processing. Spatial randomization is different because the two paths leave the
transmit array in different directions. Their array responses therefore do not
cancel in general. The phase-center agile array (PCAA) exploits this property
by randomly selecting a transmit subarray at every pulse
\cite{zheng2024pcaa,zheng2026pcaa}. The resulting path-dependent phase changes
degrade coherent integration when the pulse-wise subarray realization is not
available to the illegitimate receiver.

The associated paper extends this idea from subarray switching to general
pulse-wise spatial randomization. Its key reduction is a random complex
coefficient that contains the residual transmit-array response after
direct-reference processing. Once this coefficient is identified, the
encrypted-sensing performance of different activation mechanisms can be studied
through a common statistical metric. Independent element activation is then a
natural realization of this framework and is related to random antenna subset
selection previously considered for radar jamming \cite{wang2021rass}.

This document supplies the details omitted from the length-limited conference
paper. Section~II derives the reduced coefficient model and the finite-CPI
performance metric. Section~III proves the unified asymptotic result.
Section~IV establishes the leading finite-CPI bias and variance, including the
treatment of anomalously small sample denominators. Section~V derives the
exact coefficient moments for Bernoulli element activation.

\section{Problem Formulation}

\subsection{Random Spatial Coefficient}

Consider an $N_T$-element uniform linear transmit array with inter-element
spacing $d_T$ and wavelength $\lambda$. Its steering vector is
\begin{equation}
  \begin{aligned}
    \mathbf a_T(\theta)
    &=\left[1,e^{-j\kappa d_T\sin\theta},\ldots,
      e^{-j\kappa(N_T-1)d_T\sin\theta}\right]^T,\\
    \kappa&=\frac{2\pi}{\lambda}.
  \end{aligned}
  \label{eq:steering-vector}
\end{equation}
During pulse $n$, the array applies a random complex weight vector
$\mathbf w_n$. Let $\theta_T$ and $\theta_{TR}$ denote the transmit-to-target
and transmit-to-receiver directions, respectively, and define their spatial-frequency separation as
\begin{equation}
  \psi=\kappa d_T(\sin\theta_T-\sin\theta_{TR}).
  \label{eq:spatial-separation}
\end{equation}
The transmit-array factors on the target and direct paths are
\begin{equation}
  u_{T,n}=\mathbf a_T^T(\theta_T)\mathbf w_n,
  \qquad
  u_{D,n}=\mathbf a_T^T(\theta_{TR})\mathbf w_n.
  \label{eq:path-factors}
\end{equation}
After the target echo is correlated with the direct-path reference, all
transmit-side spatial dependence is contained in
\begin{equation}
  g_n=u_{T,n}u_{D,n}^*
  =\left[\mathbf a_T^T(\theta_T)\mathbf w_n\right]
   \left[\mathbf a_T^T(\theta_{TR})\mathbf w_n\right]^*.
  \label{eq:spatial-coefficient}
\end{equation}
Thus, temporal or carrier phase common to the two received paths is removed,
whereas the differential spatial response remains. If the weights are drawn
independently from the same activation law at every pulse, then
$g_0,g_1,\ldots$ are independent and identically distributed (i.i.d.).

\subsection{Finite-CPI Performance Metric}

At the correct target parameters, a legitimate receiver that knows the
activation sequence can include $\{g_n\}$ in its matched template. Its output
signal power is proportional to $\sum_{n=0}^{N-1}|g_n|^2$. In contrast, an
illegitimate receiver without the pulse-wise activation realization uses the
conventional angle--Doppler template, whose output signal power is
proportional to $N^{-1}|\sum_{n=0}^{N-1}g_n|^2$. Under identical normalized
output-noise power, their output-SNR ratio is the SNRL used in
\cite{zheng2026pcaa}:

For every $N$ for which
$\sum_{n=0}^{N-1}|g_n|^2>0$, define
\begin{equation}
  \rho_N=
  \frac{\left|\sum_{n=0}^{N-1}g_n\right|^2}
       {N\sum_{n=0}^{N-1}|g_n|^2}.
  \label{eq:rhoN}
\end{equation}
On the event $\sum_{n=0}^{N-1}|g_n|^2=0$, set $\rho_N=0$.
This convention also covers a CPI in which all elements remain inactive.
By the Cauchy--Schwarz inequality, $0\leq\rho_N\leq1$. A value near one means
that the illegitimate receiver retains almost all coherent-processing gain, so
encrypted sensing is nearly ineffective; a smaller value indicates a larger
coherent-processing loss relative to the legitimate receiver.

\section{Unified Asymptotic Performance}

\begin{proposition}[Unified asymptotic SNRL]
Suppose that $\{g_n\}$ is i.i.d. with $\mu_g=\E[g_n]$ and
$0<q_g=\E[|g_n|^2]<\infty$. Then
\begin{equation}
  \rho_N\xrightarrow{\mathrm{a.s.}}
  \rho_\infty=\frac{|\mu_g|^2}{q_g}
  =\frac{|\mu_g|^2}{|\mu_g|^2+\sigma_g^2},
  \qquad
  \sigma_g^2=\E[|g_n-\mu_g|^2],
  \label{eq:rhoinf}
\end{equation}
and $\E[\rho_N]\to\rho_\infty$.
\end{proposition}

\begin{proof}
Introduce the sample means
\begin{equation}
  \bar g_N=\frac{1}{N}\sum_{n=0}^{N-1}g_n,
  \qquad
  \bar q_N=\frac{1}{N}\sum_{n=0}^{N-1}|g_n|^2.
  \label{eq:samplemeans}
\end{equation}
Because $\E|g_n|\leq\sqrt{\E|g_n|^2}<\infty$, the strong law of large
numbers applies to both $g_n$ and $|g_n|^2$, and therefore
\begin{equation}
  \bar g_N\xrightarrow{\mathrm{a.s.}}\mu_g,
  \qquad
  \bar q_N\xrightarrow{\mathrm{a.s.}}q_g.
  \label{eq:slln}
\end{equation}
Since $q_g>0$, the denominator $\bar q_N$ is positive for all sufficiently
large $N$ on an event of probability one.  Rewriting \eqref{eq:rhoN} as
\begin{equation}
  \rho_N=\frac{|\bar g_N|^2}{\bar q_N}
  \label{eq:ratioform}
\end{equation}
and applying the continuous mapping theorem to \eqref{eq:slln} gives
\begin{equation}
  \rho_N\xrightarrow{\mathrm{a.s.}}\frac{|\mu_g|^2}{q_g}.
\end{equation}
Furthermore,
\begin{equation}
  q_g=\E|g_n|^2
  =|\E g_n|^2+\E|g_n-\E g_n|^2
  =|\mu_g|^2+\sigma_g^2,
\end{equation}
which proves the second equality in \eqref{eq:rhoinf}.

It remains to prove convergence of the expectation.  By the
Cauchy--Schwarz inequality,
\begin{equation}
  \left|\sum_{n=0}^{N-1}g_n\right|^2
  \leq N\sum_{n=0}^{N-1}|g_n|^2,
\end{equation}
so $0\leq\rho_N\leq1$.  The sequence $\{\rho_N\}$ is thus uniformly
bounded.  Almost-sure convergence together with the bounded convergence
theorem yields
\begin{equation}
  \lim_{N\to\infty}\E[\rho_N]
  =\E\!\left[\lim_{N\to\infty}\rho_N\right]
  =\rho_\infty.
\end{equation}
\end{proof}

\section{Finite-CPI Performance}

For the finite-$N$ expansion below, assume that $g_n$ has finite support
and is not identically zero.  These conditions ensure the required moment
bounds and make the probability of a vanishing or anomalously small sample
denominator exponentially small.

Define
\begin{equation}
  \begin{aligned}
    \mathbf x_n
    &=\begin{bmatrix}
      \Re\{g_n\} & \Im\{g_n\} & |g_n|^2
    \end{bmatrix}^{T},\\
    \mathbf m&=\E[\mathbf x_n],\qquad
    \mathbf\Sigma=\Cov(\mathbf x_n),
  \end{aligned}
  \label{eq:xdef}
\end{equation}
and
\begin{equation}
  f(\mathbf x)=\frac{x_1^2+x_2^2}{x_3},\qquad x_3>0.
  \label{eq:fdef}
\end{equation}

\begin{proposition}[Finite-CPI SNRL]
Under the preceding assumptions,
\begin{equation}
  \E[\rho_N]=\rho_\infty+\frac{B}{N}+o(N^{-1}),
  \qquad
  \Var(\rho_N)=\frac{V}{N}+o(N^{-1}),
  \label{eq:finiteexpansion}
\end{equation}
where
\begin{equation}
  B=\frac12\operatorname{tr}
  \!\left[\nabla^2 f(\mathbf m)\mathbf\Sigma\right],
  \qquad
  V=\nabla f(\mathbf m)^T\mathbf\Sigma\nabla f(\mathbf m).
  \label{eq:BV}
\end{equation}
\end{proposition}

\begin{proof}
Write
\begin{equation}
  \bar{\mathbf x}_N=\frac1N\sum_{n=0}^{N-1}\mathbf x_n,
  \qquad
  \mathbf h_N=\bar{\mathbf x}_N-\mathbf m.
  \label{eq:hdef}
\end{equation}
Whenever the third component of $\bar{\mathbf x}_N$ is positive,
\begin{equation}
  \rho_N=f(\bar{\mathbf x}_N)=f(\mathbf m+\mathbf h_N).
  \label{eq:rhof}
\end{equation}
The first two moments of $\mathbf h_N$ are
\begin{equation}
  \E[\mathbf h_N]=\mathbf0,
  \qquad
  \E[\mathbf h_N\mathbf h_N^T]=\frac{\mathbf\Sigma}{N}.
  \label{eq:hmoments}
\end{equation}

Let $m_3=q_g>0$.  Since $|g_n|^2$ is bounded, Hoeffding's inequality gives
constants $c_1,c_2>0$ such that
\begin{equation}
  \Pr\!\left(\bar x_{N,3}<\frac{q_g}{2}\right)
  \leq c_1e^{-c_2N}.
  \label{eq:lower-tail}
\end{equation}
Define $R_N$ for all realizations through the identity
\begin{equation}
  \rho_N=f(\mathbf m)
  +\nabla f(\mathbf m)^T\mathbf h_N
  +\frac12\mathbf h_N^T\nabla^2f(\mathbf m)\mathbf h_N
  +R_N,
  \label{eq:taylor}
\end{equation}
On the complementary event in \eqref{eq:lower-tail}, $f$ has bounded
derivatives through third order on the line segment joining $\mathbf m$ and
$\bar{\mathbf x}_N$. The second-order multivariate Taylor formula therefore
gives $|R_N|\leq C\|\mathbf h_N\|^3$ there, for a constant $C$ independent
of $N$. On the exceptional event, boundedness of $\mathbf x_n$ and
$0\leq\rho_N\leq1$ imply that $R_N$ is uniformly bounded. Standard moment
bounds for sums of centered independent random vectors give
\begin{equation}
  \E\|\mathbf h_N\|^3=O(N^{-3/2}),
  \qquad
  \E\|\mathbf h_N\|^4=O(N^{-2}).
  \label{eq:remainder-moments}
\end{equation}
The contribution of the exceptional event in \eqref{eq:lower-tail} is
therefore $O(e^{-c_2N})$. Consequently,
\begin{equation}
  \E|R_N|=o(N^{-1}).
  \label{eq:remaindermean}
\end{equation}

Taking expectations in \eqref{eq:taylor}, using \eqref{eq:hmoments}, and
applying
\begin{equation}
  \E[\mathbf h_N^T\mathbf H\mathbf h_N]
  =\operatorname{tr}\!\left(
    \mathbf H\E[\mathbf h_N\mathbf h_N^T]
  \right)
\end{equation}
with $\mathbf H=\nabla^2f(\mathbf m)$ gives
\begin{align}
  \E[\rho_N]
  &=f(\mathbf m)
    +\frac{1}{2N}\operatorname{tr}
      [\nabla^2f(\mathbf m)\mathbf\Sigma]
    +o(N^{-1}) \\
  &=\rho_\infty+\frac{B}{N}+o(N^{-1}).
  \label{eq:meanresult}
\end{align}

For the variance, set
\begin{equation}
  L_N=\nabla f(\mathbf m)^T\mathbf h_N,
  \qquad
  Q_N=\rho_N-f(\mathbf m)-L_N.
\end{equation}
Equations \eqref{eq:taylor}, \eqref{eq:lower-tail}, and
\eqref{eq:remainder-moments} imply
\begin{equation}
  \E[Q_N^2]=O(N^{-2})=o(N^{-1}).
  \label{eq:qbound}
\end{equation}
Moreover, $\E[L_N]=0$ and
\begin{equation}
  \Var(L_N)
  =\nabla f(\mathbf m)^T
    \E[\mathbf h_N\mathbf h_N^T]
    \nabla f(\mathbf m)
  =\frac{V}{N}.
  \label{eq:linearvariance}
\end{equation}
By Cauchy--Schwarz,
\begin{equation}
  |\Cov(L_N,Q_N)|
  \leq\sqrt{\Var(L_N)\Var(Q_N)}=o(N^{-1}).
\end{equation}
It follows that
\begin{align}
  \Var(\rho_N)
  &=\Var(L_N)+2\Cov(L_N,Q_N)+\Var(Q_N)\\
  &=\frac{V}{N}+o(N^{-1}),
\end{align}
which completes the proof.
\end{proof}

For completeness, if $\mathbf m=(a,b,q_g)^T$, direct differentiation of
\eqref{eq:fdef} gives
\begin{equation}
  \nabla f(\mathbf m)=
  \begin{bmatrix}
    2a/q_g\\[2pt]
    2b/q_g\\[2pt]
    -(a^2+b^2)/q_g^2
  \end{bmatrix}
  \label{eq:gradient}
\end{equation}
and
\begin{equation}
  \nabla^2f(\mathbf m)=
  \begin{bmatrix}
    2/q_g & 0 & -2a/q_g^2\\
    0 & 2/q_g & -2b/q_g^2\\
    -2a/q_g^2 & -2b/q_g^2 & 2(a^2+b^2)/q_g^3
  \end{bmatrix}.
  \label{eq:hessian}
\end{equation}
Substitution of \eqref{eq:gradient}--\eqref{eq:hessian} into
\eqref{eq:BV} gives the constants $B$ and $V$ appearing in
\eqref{eq:finiteexpansion}.

\section{Bernoulli Element Activation}

Let $0<p\leq1$, and let $b_1,\ldots,b_{N_T}$ be independent
$\operatorname{Bernoulli}(p)$ variables. The transmit weights are
$\mathbf w=\mathbf b\odot\mathbf a_T^*(\theta_T)$, where
$\mathbf b=[b_1,\ldots,b_{N_T}]^T$, so every active element is phase aligned
toward the target direction. Define
\begin{equation}
  z_m=e^{-j(m-1)\psi},\qquad
  K=\sum_{i=1}^{N_T}b_i,\qquad
  A(\psi)=\sum_{m=1}^{N_T}z_m.
\end{equation}
The pulse index is omitted because the coefficients are identically
distributed over pulses.  The Bernoulli spatial coefficient is
\begin{equation}
  g=K\sum_{m=1}^{N_T}b_mz_m.
  \label{eq:bernoulli-g}
\end{equation}
Define
\begin{align}
  D_p&=p[1+3(N_T-1)p+(N_T-1)(N_T-2)p^2],\label{eq:Dp}\\
  O_p&=p^2[4+5(N_T-2)p+(N_T-2)(N_T-3)p^2].\label{eq:Op}
\end{align}

\begin{proposition}[Bernoulli SNRL]
The first two moments of \eqref{eq:bernoulli-g} are
\begin{align}
  \mu_g^{(B)}&=[p+(N_T-1)p^2]A(\psi),\label{eq:bernoulli-mean}\\
  q_g^{(B)}&=N_T D_p+(|A(\psi)|^2-N_T)O_p.\label{eq:bernoulli-second}
\end{align}
Whenever $q_g^{(B)}>0$, the corresponding asymptotic SNRL is
\begin{equation}
  \rho_\infty^{(B)}=
  \frac{[p+(N_T-1)p^2]^2|A(\psi)|^2}
       {N_T D_p+(|A(\psi)|^2-N_T)O_p}.
  \label{eq:bernoulli-rho}
\end{equation}
\end{proposition}

\begin{proof}
Expanding \eqref{eq:bernoulli-g} gives
\begin{equation}
  g=\sum_{i=1}^{N_T}\sum_{m=1}^{N_T}b_ib_mz_m.
  \label{eq:g-expanded}
\end{equation}
Since $b_i^2=b_i$, independence yields
\begin{equation}
  \E[b_ib_m]=
  \begin{cases}
    p, & i=m,\\
    p^2, & i\ne m.
  \end{cases}
\end{equation}
Therefore,
\begin{align}
  \E[g]
  &=\sum_{m=1}^{N_T}z_m
    \left(\E[b_m]+\sum_{i\ne m}\E[b_i b_m]\right)\\
  &=[p+(N_T-1)p^2]\sum_{m=1}^{N_T}z_m,
\end{align}
which proves \eqref{eq:bernoulli-mean}.

For the second moment,
\begin{align}
  |g|^2
  &=K^2\left|\sum_{m=1}^{N_T}b_mz_m\right|^2\\
  &=\sum_{m=1}^{N_T}K^2b_m
    +\sum_{\substack{m,l=1\\m\ne l}}^{N_T}
      K^2b_mb_lz_mz_l^*.
  \label{eq:q-expanded}
\end{align}
The expectations multiplying all diagonal terms are identical.  For a
fixed $m$, write $K=b_m+T_m$, where
$T_m\sim\operatorname{Binomial}(N_T-1,p)$ is independent of $b_m$. Hence
\begin{align}
  \E[K^2b_m]
  &=p\E[(1+T_m)^2]\\
  &=p\{1+2(N_T-1)p+(N_T-1)p(1-p)\notag\\
  &\hspace{3.2em}{}+(N_T-1)^2p^2\}\\
  &=D_p.
  \label{eq:diagonal-expectation}
\end{align}

Similarly, for fixed $m\ne l$, write
$K=b_m+b_l+T_{ml}$, where
$T_{ml}\sim\operatorname{Binomial}(N_T-2,p)$. Conditioning on
$b_m=b_l=1$ gives
\begin{align}
  \E[K^2b_mb_l]
  &=p^2\E[(2+T_{ml})^2]\\
  &=p^2\{4+4(N_T-2)p+(N_T-2)p(1-p)\notag\\
  &\hspace{3.2em}{}+(N_T-2)^2p^2\}\\
  &=O_p.
  \label{eq:offdiagonal-expectation}
\end{align}
Taking expectations in \eqref{eq:q-expanded} and using
\begin{equation}
  \sum_{\substack{m,l=1\\m\ne l}}^{N_T}z_mz_l^*
  =\left|\sum_{m=1}^{N_T}z_m\right|^2-N_T
  =|A(\psi)|^2-N_T
\end{equation}
yields
\begin{equation}
  \E|g|^2=N_T D_p+(|A(\psi)|^2-N_T)O_p,
\end{equation}
which proves \eqref{eq:bernoulli-second}.  Finally, Proposition 1 gives
\begin{equation}
  \rho_\infty^{(B)}
  =\frac{|\mu_g^{(B)}|^2}{q_g^{(B)}},
\end{equation}
and substitution of \eqref{eq:bernoulli-mean} and
\eqref{eq:bernoulli-second} establishes \eqref{eq:bernoulli-rho}.
\end{proof}

\IEEEtriggeratref{4}

\end{document}